\documentclass{article}
\usepackage{csquotes}
\usepackage[english]{babel}
\usepackage{listings}
\usepackage{xcolor}

\usepackage[letterpaper,top=2cm,bottom=2cm,left=3cm,right=3cm,marginparwidth=1.75cm]{geometry}
\usepackage[
  backend=biber,
  style=alphabetic,   
  sorting=nyt
]{biblatex}

\usepackage{mdframed}
\usepackage{amsmath}
\usepackage{amssymb}
\usepackage{amsmath,amssymb}
\usepackage[most]{tcolorbox}
\usepackage{thmtools}
\usepackage{amsthm}
\usepackage{bbm}
\usepackage{booktabs}   
\usepackage{graphicx}
\usepackage{braket}
\usepackage{comment}
\usepackage{xcolor}
\usepackage{tikz-cd}
\usepackage[colorlinks=true, allcolors=blue]{hyperref}
\usepackage{changepage}
\newtheoremstyle{indented}
  {\topsep}   
  {\topsep}   
  {\itshape}  
  {0em}     
  {\bfseries} 
  {.}         
  {.5em}      
  {}          
\theoremstyle{indented}

\newtheorem{theorem}{Theorem}[section] 
\newtheorem*{theorem*}{Theorem}

\newtheorem*{corollary*}{Corollary}
\newtheorem{lemma}[theorem]{Lemma} 

\newtheorem{proposition}[theorem]{Proposition} 
\usepackage{caption}
\usepackage{enumitem}
\usepackage{mathtools}

\usepackage{thm-restate}
\usepackage{multirow}

\newcommand{\tr}{\operatorname{tr}}
\newcommand{\diag}{\operatorname{diag}}

\title{Maximally entangled states are not complete for pseudo-telepathy}
\author{Olivier Lalonde \\
\small Institute for Quantum Computing \\
\small University of Waterloo \\
\small \texttt{olalonde@uwaterloo.ca}
}
\date{}

\begin{document}
\maketitle

\begin{abstract}
    One of the longstanding open problems in quantum nonlocality is to determine if maximally entangled states are complete for bipartite pseudo-telepathic games: namely, if every nonlocal game which admits a perfect entangled strategy admits such a strategy which uses a maximally entangled state. We exhibit a counterexample to this in the form of a bipartite nonlocal game with input sets of size 4 and 3 and output sets both of size 6. This game is part of a new class of nonlocal games, which we call inner product games, which could be of independent interest.
\end{abstract}

\section{Introduction}
Quantum pseudo-telepathy (\cite{heywood1983nonlocality, Brassard_2005}) refers to the phenomenon in which a given nonlocal game (\cite{cleve2010consequenceslimitsnonlocalstrategies}) admits a perfect entangled strategy but no perfect classical strategy. A game exhibiting this property is said to be pseudo-telepathic, and the most famous examples of such games are perhaps the magic square game (\cite{mermin1990simple, peres1990incompatible, aravind2002bell}) for the bipartite setting, and the Mermin-GHZ games (\cite{ghz1989going}, \cite{ghsz1990bell}, \cite{mermin1990mysteries}) for the multipartite setting. 

A longstanding open problem in this area is to characterize the type of entanglement that can be required by perfect entangled strategies for pseudo-telepathic games. Although a number of families of bipartite pseudo-telepathic games are known, it turns out that maximally entangled states always suffice to play these games perfectly, and it is open to determine if this is the case in general (\cite{mancinska2015maximallyentangledstatespseudotelepathy}). It should be noted that it has been known for some time that, for general nonlocal games, restricting to strategies using a maximally entangled state can lower the value of the game (\cite{Junge_2011, Vidick_2011}). Some partial progress in this direction was made recently by \cite{renner2026pureentangledstateslead}, which showed the existence of pseudo-telepathic games which can be won perfectly using a variety of non-maximally entangled states. This shows that maximally entangled states are not necessary to execute perfect entangled strategies for pseudo-telepathic games in general, but does not address the problem of whether they are always sufficient. 

In this work, we settle this question in the negative by exhibiting a counterexample:
\begin{theorem} \label{thm:thetheorem}
    There exists a bipartite game $\frak{G}_{4,3,6,6}$ with input sets of size 4 and 3 and output sets both of size 6 which has a perfect entangled strategy involving an entangled state of local dimension 6, but no perfect entangled strategy involving a maximally entangled state.  
\end{theorem}
The game $\frak{G}_{4,3,6,6}$ is an instance of a new class of nonlocal games which we call inner product games and which are introduced in section \ref{sec:ipgames}. The construction of $\frak{G}_{4,3,6,6}$, which was found by computer search, as well as the proof that no perfect entangled strategy involving a maximally entangled state exists for it, which is also numerical and is based on the tracial NPA hierarchy, is described in section \ref{sec:thegame}. 

\section{Preliminaries}
Given two finite quantum systems $A$ and $B$ of dimension $m$, the (unnormalized) standard maximally entangled state on $AB$ is defined by:
\[\ket{\Omega}_{AB} = \sum_{i \in [m]} \ket{i}_A \ket{i}_B\]
More generally, a pure state on $AB$ is said to be maximally entangled if it has full Schmidt spectrum and its Schmidt coefficients are all equal. It follows from the Schmidt decomposition theorem that such a state is equivalent to $\ket{\Omega}$ up to local unitaries. The following properties of $\ket{\Omega}_{AB}$ are standard:
\begin{lemma} \label{lem:Omega}
    We have, for any square matrices $M,N$ of order $m$: 
    \begin{align*}
        (M_A \otimes I_B)\ket{\Omega}_{AB} &= (I_A \otimes M^T_B) \ket{\Omega}_{AB}\\
        \braket{\Omega | (M_A \otimes N_B) | \Omega} &= \tr(MN^T)
    \end{align*}
\end{lemma}

A bipartite nonlocal game $\frak{G} = (X,Y,A,B,W, \pi)$ is specified by finite input sets $X,Y$, finite output sets $A,B$, a subset $W \subseteq X \times Y \times A \times B$ and a probability distribution $\pi$ on $X \times Y$. The associated scenario is that Alice and Bob are spatially separated, given inputs $x \in X, y \in Y$ sampled according to the distribution $\pi$, and requested to produce outputs $a \in A, b \in B$ according to some strategy. They are said to have won the game if $(x,y,a,b) \in W$. A strategy is said to be perfect if this is achieved with probability one. A strategy is said to be deterministic if Alice and Bob's outputs are deterministic functions of their inputs, and it is said to be classical if it is a probabilistic combination of deterministic strategies. It is simple to see that a perfect classical strategy exists for the game if and only if a perfect deterministic strategy exists for it. A strategy is said to be entangled if Alice and Bob share a given entangled state and their outputs are produced by measuring their respective share of the state. It is said to be maximally entangled if said shared entangled state is maximally entangled. Since any maximally entangled state $\ket{\Phi}_{AB}$ is equivalent to $\ket{\Omega}_{AB}$ up to local unitaries, it may be assumed that Alice and Bob share $\ket{\Omega}_{AB}$ in a maximally entangled strategy. 

When one is only interested in the existence of a perfect strategy for the game in a given model, the choice of probability distribution $\pi$ is irrelevant provided that every input pair has nonzero probability, and will be taken to be the uniform distribution unless specified otherwise. The game is said to be pseudo-telepathic if it admits a perfect entangled strategy but no perfect classical strategy. The first game with this property to have been introduced is due to \cite{heywood1983nonlocality}. 

A $*$-algebra $\mathcal{A}$ is an associative $\mathbb{C}$-algebra equipped with a conjugate-linear involution $*$ satisfying $(xy)^* = y^*x^*$ for all $x,y \in \mathcal{A}$. A $*$-algebra is said to be unital if it has a multiplicative identity, denoted 1: in this work, all $*$-algebras are assumed to be unital. Given a set $X$, we will write $\mathbb{C}\langle  X, X^*\rangle$ to mean the free unital $*$-algebra generated by $X$, i.e. the set of noncommutative polynomials in the elements of $X$ and their formal adjoints. An element $x \in \mathcal{A}$ is said to be positive if $x = y^* y$ for some $y \in \mathcal{A}$. A linear functional $f: \mathcal{A} \to \mathbb{C}$ is said to be positive if $f(x) \geq 0$ for all positive $x$. A positive linear functional $f$ is said to be a state if $f(1) = 1$, and a state $f$ is said to be tracial if $f(xy) = f(yx)$ for all $x,y \in \mathcal{A}$. 

\section{Inner product games} \label{sec:ipgames}
We introduce a new class of nonlocal games, which we call inner product games. These games are defined as follows. Let $m \geq 2$, let $\{U_x\}_{x \in [n_X]}$, $\{V_y\}_{y \in [n_Y]}$ be collections of unitary operators of order $m$, and let $S \neq 0$ be a positive semidefinite matrix of order $m$. Alice and Bob will be given $x \in [n_X], y \in [n_Y]$, respectively, and are to output $a, b \in [m]$. The winning condition is:
\begin{equation} \label{eq:winningcondition} \bra{a} U_x^\dagger S V_y \ket{b} \neq 0 \end{equation}
The justification for the name is that one may think of Alice and Bob as being given orthonormal bases of $\mathbb{C}^m$, here represented by the columns of $U_x$ and $V_y$, and having to produce vectors in their respective bases which are not orthogonal. We note that this class of games is quite general, and, for example, subsumes the games of \cite{renner2004pseudotelepathy, Cabello_2025}, which are recovered in the special case where $S = I$. We show:
\begin{proposition} \label{prop:perfectIP}
    Any inner product game has a perfect entangled strategy with an entangled state of local dimension $m$. 
\end{proposition}
\begin{proof}
    We can assume without loss of generality that $\tr(S^2) = 1$ by rescaling $S$ if needed, which doesn't affect the game. Set:
    \[\ket{\psi} = (S \otimes I) \ket{\Omega} \]
    The second point of Lemma \ref{lem:Omega} shows that $\ket{\psi}$ is a unit vector. We describe an entangled strategy which uses $\ket{\psi}$ as the shared entangled state. Given inputs $x,y$, Alice and Bob respectively apply $U_x^\dagger$ and $V_y^T$ to their shares of the state and measure in the computational basis to obtain their outputs $a, b$. We work out:
    \begin{align*}
        P(a,b \mid x,y) &= \bra{\psi} ((U_x \ket{a} \bra{a} U_x^\dagger) \otimes (\overline{V_Y} \ket{b} \bra{b} V_y^T)) \ket{\psi}\\
        &= \bra{\Omega} (S(U_x \ket{a} \bra{a} U_x^\dagger S) \otimes (\overline{V_Y} \ket{b} \bra{b} V_y^T)) \ket{\Omega}\\
        &= \tr(S(U_x \ket{a} \bra{a} U_x^\dagger S V_y \ket{b} \bra{b} V_y^\dagger)\\
        &= \|\bra{a} U_x^\dagger S V_y \ket{b}\|^2
    \end{align*}
    Therefore, if $(a,b)$ is a losing answer given the inputs $x, y$, its probability of occurring is zero. Thus this is a perfect entangled strategy.
\end{proof}

\section{Proving the main theorem} \label{sec:thegame}

Having introduced inner product games in the previous section, we turn to proving theorem \ref{thm:thetheorem}. Subsection \ref{subsec:gamedescription} describes the game $\frak{G}_{4,3,6,6}$ mentioned in the statement of the theorem and the algorithm we used to construct it. Subsection \ref{subsec:nome} describes the proof that no perfect maximally entangled strategy exists for $\frak{G}_{4,3,6,6}$.

\subsection{The case of the game $\frak{G}_{4,3,6,6}$} \label{subsec:gamedescription}
The game $\frak{G}_{4,3,6,6}$ is described as follows. We will take $m=6, n_X = 4, n_Y = 3$. The positive semidefinite matrix $S$ will be taken to be:
\[S = \diag(1,1,2,2,2,2)\]
The unitaries $\{U_x\}_{x \in [4]}$ and $\{V_y\}_{y \in [3]}$ which define the game are given in appendix \ref{app:unitaries}. By the proof of proposition \ref{prop:perfectIP}, this game has a perfect entangled strategy using the following entangled state:
\begin{equation} \label{eq:thestate}
    \ket{\Phi} = \frac{1}{\sqrt{18}}\left(\ket{1}\ket{1} + \ket{2} \ket{2} + 2\ket{3}\ket{3} + 2\ket{4}\ket{4} + 2\ket{5}\ket{5} + 2\ket{6}\ket{6} \right)
\end{equation}

This game was generated by means of a computer search using a heuristic procedure analogous to the one used by \cite{lalonde2025quantum} to manufacture the graph $G_{21}$.   For a given $m \geq 3$, all unitaries $\mathcal{U}$ in $U(m)$ whose columns all have components in $\{-1,0,1\}$ before normalization were enumerated up to trivial symmetries, such as permuting the columns or multiplying a given column by $-1$. We then searched for a non-scalar diagonal matrix $S$ with integral entries such that the corresponding inner product game with possible Alice and Bob input unitaries $\mathcal{U}$ is classically unfeasible. These choices were made so that the condition (\ref{eq:winningcondition}) is not satisfied for a non-negligible proportion of $x,y,a,b$, which heuristically increases the likelihood of the resulting game being classically unfeasible. This was unsuccessfully attempted for $m \in \{3,4,5\}$, and finally succeeded for $m=6$. Once our choice of $S$ was identified, a simple pruning strategy consisting of successively deleting an input whose removal does not make the game classically feasible until this is no longer possible was performed several times, and the smallest game obtained in this way was kept. 

It is simple to verify using exhaustive search that the game $\frak{G}_{4,3,6,6}$ has no perfect classical strategy. Much more strongly, numerical simulations suggest that the game is a self-test for the state (\ref{eq:thestate}), i.e. any entangled state $\ket{\Phi'}$ which enables a perfect entangled strategy for the game is such that $\ket{\Phi}$ can be extracted out of $\ket{\Phi'}$ using local operations. Proving this analytically seems difficult, but it is possible to show that no perfect maximally entangled strategy for the game exists. 

\subsection{Ruling out perfect maximally entangled strategies} \label{subsec:nome}
In order to rule out the existence of a perfect maximally entangled strategy for the game $\frak{G}_{4,3,6,6}$, we will use a variant of the sum-of-squares hierarchy of \cite{Navascu_s_2008, doherty2008quantummomentproblembounds}. Although, in its original form, this hierarchy is applicable to general entangled strategies, it is possible to impose additional constraints which are specifically satisfied by ones using a maximally entangled state, in a manner analogous to what was done in \cite{russell2023synchronousnpahierarchyapplications}. Establishing that there exists no correlation in one of the resulting relaxed correlation sets which wins at the game with probability one completes the proof of theorem \ref{thm:thetheorem}.

Let $\{M^x_a\}_{x \in [4], a \in [6]}$ and $\{N^y_b\}_{y \in [3], b \in [6]}$ be formal generators. Set $E^x_a = (M^x_a)^* M^x_a, F^y_b = (N^y_b)^* N^y_b$, and let $\mathcal{A}$ be the $*$-algebra generated by the $\{M^x_a\}$ and $\{N^y_b\}$ subject to the constraints that $[E^x_a, F^y_b] = 0$ for all $x,y,a,b$ and, for all $x \in [4], y \in [3]$, 
\[\sum_a E^x_a = \sum_b F^y_b = 1\]
It follows from part two of lemma \ref{lem:Omega} that the existence of a perfect maximally entangled strategy for $\frak{G}_{4,3,6,6}$ would imply the existence of a tracial state $f$ on $\mathcal{A}$ which satisfies $f(E^x_a F^y_b) = 0$ whenever $a,b$ are losing outputs given the inputs $x,y$.  Although determining whether such a state exists is undecidable in general (\cite{Slofstra_2019}), there is a hierarchy of relaxations of this problem which are tractable via semi-definite programming, which correspond to testing for the existence of a pseudo-state defined only on words in the generators of length bounded by a parameter $n$. The inexistence of such a pseudo-state implies the inexistence of such a global state. 

Using a slightly strengthened version of the $n=4$ level of this hierarchy, an AI agent (Codex) was able to demonstrate the inexistence of a tracial state of the aforementioned form, thereby completing the proof of Theorem \ref{thm:thetheorem}. The resulting semidefinite program involved $207,202$ scalar variables, $208,702$ constraints and required 7 minutes to solve on a laptop. A rational certificate of infeasibility was extracted which takes up 5.6 megabytes of memory. The entire proof of nonexistence of a perfect maximally entangled strategy for $\frak{G}_{4,3,6,6}$ was formalized in the proof assistant Lean (\cite{demoura2021lean4}), as \texttt{theorem noPerfectMaximallyEntangledStrategy (d : Nat)}. The proof package can be found in the Github repository \url{https://github.com/lalondeo/MaximallyEntangledIncomplete}.

\section{Conclusion}
In this work, we have settled the longstanding problem of whether maximally entangled states are always sufficient for pseudo-telepathy in the negative by introducing a counterexample in the form of the game $\frak{G}_{4,3,6,6}$. Our work suggests a number of avenues for future research, most prominently:
\begin{enumerate}
    \item Our resolution of the problem is heavily computer-assisted, and our counterexample was manufactured to have the desired properties. It would be interesting to look for a more natural counterexample, for which the insufficiency of maximally entangled states could perhaps be established analytically rather than by using semidefinite hierarchies, which yield little insight into why the result is true. 
    \item As mentioned in subsection \ref{subsec:gamedescription}, numerics suggest that the game $\frak{G}_{4,3,6,6}$ is a self-test for the entangled state (\ref{eq:thestate}). It would be interesting to try to prove this formally. In a more general vein, it is known (\cite{brassard2004minimumentangledstatedimension}, \cite{renner2004pseudotelepathy}) that unlike for nonlocality, which is enabled by any entangled pure state (\cite{gisin1992maximal}), there exist entangled pure states which do not enable pseudo-telepathy in any dimension. It would be interesting to study the set of states which admit a pseudo-telepathic self-test. Could this set be strictly smaller than the set of states which enable pseudo-telepathy? 
    \item Our result could also have implications for zero-error quantum information theory. A graph parameter called the entangled chromatic number was introduced in \cite{Briet_2015} in connection to the zero-error source-channel coding problem, and it was shown there that this graph parameter is upper bounded by the quantum chromatic number (\cite{cameron2006quantumchromaticnumbergraph}) and that the existence of a separation between the two would imply the existence of a pseudo-telepathic game for which maximally entangled states are insufficient for a perfect entangled strategy. It would be interesting to see if the reverse implication can be established, in which case our game $\frak{G}_{4,3,6,6}$ could be used to yield such a separation.
\end{enumerate}

\section{Acknowledgments}
We thank William Slofstra for useful exchanges, as well as Arthur Mehta for pointing us to the work \cite{renner2026pureentangledstateslead}.  We acknowledge the support of the Natural Sciences and Engineering Research Council of Canada (NSERC) grants ALLRP-578455-2022 and RGPIN-2023-03731. 
\paragraph{Statement of AI use}
AI tools were used extensively during the course of this work, but only as a means of executing ideas rather than producing them. Inner product games, the heuristic approach that was used to generate the game $\frak{G}_{4,3,6,6}$ as well as the idea of modifying the hierarchy of \cite{russell2023synchronousnpahierarchyapplications} to rule out the existence of a maximally entangled strategy were all invented by the author while very bored during a long car trip. The AI agent Codex was used to search for a matrix $A$ such that the resulting game is classically unfeasible as well as to implement the pruning algorithm to reduce the input sizes. On the author's request, Codex also implemented and ran the tracial NPA hierarchy, extracted a rational unfeasibility certificate, and formalized the inexistence of a perfect entangled strategy using a maximally entangled state for the game in Lean. The paper was written exclusively by the author, with AI tools only being used to search the literature. 

\printbibliography

\appendix

\section{The unitaries defining the game $\frak{G}_{4,3,6,6}$}
\label{app:unitaries}
The unitaries of order 6 which define the inner product game $\frak{G}_{4,3,6,6}$ are the following.
\[
\setlength{\arraycolsep}{3pt}
\renewcommand{\arraystretch}{1.15}
\begin{aligned}
U_1 &=
\begin{pmatrix}
0 & 0 & 0 & 0 & \tfrac{1}{\sqrt{2}} & \tfrac{1}{\sqrt{2}} \\
0 & 0 & 0 & 0 & -\tfrac{1}{\sqrt{2}} & \tfrac{1}{\sqrt{2}} \\
0 & 0 & 0 & 1 & 0 & 0 \\
0 & 0 & 1 & 0 & 0 & 0 \\
0 & 1 & 0 & 0 & 0 & 0 \\
1 & 0 & 0 & 0 & 0 & 0
\end{pmatrix},
&\qquad
U_2 &=
\begin{pmatrix}
\tfrac{1}{\sqrt{2}} & \tfrac{1}{\sqrt{2}} & 0 & 0 & 0 & 0 \\
-\tfrac{1}{\sqrt{2}} & \tfrac{1}{\sqrt{2}} & 0 & 0 & 0 & 0 \\
0 & 0 & \tfrac{1}{2} & \tfrac{1}{2} & \tfrac{1}{2} & \tfrac{1}{2} \\
0 & 0 & -\tfrac{1}{2} & -\tfrac{1}{2} & \tfrac{1}{2} & \tfrac{1}{2} \\
0 & 0 & -\tfrac{1}{2} & \tfrac{1}{2} & -\tfrac{1}{2} & \tfrac{1}{2} \\
0 & 0 & \tfrac{1}{2} & -\tfrac{1}{2} & -\tfrac{1}{2} & \tfrac{1}{2}
\end{pmatrix}
\\[1.5em]
U_3 &= \frac{1}{2}
\begin{pmatrix}
0 & 0 & 1 & 1 & 1 & 1 \\
0 & 0 & -1 & -1 & 1 & 1 \\
1 & 1 & 0 & 0 & -1 & 1 \\
-1 & -1 & 0 & 0 & -1 & 1 \\
-1 & 1 & -1 & 1 & 0 & 0 \\
-1 & 1 & 1 & -1 & 0 & 0
\end{pmatrix},
&\qquad
U_4 &= \frac{1}{2}
\begin{pmatrix}
0 & 0 & 1 & 1 & 1 & 1 \\
0 & 0 & -1 & -1 & 1 & 1 \\
1 & 1 & -1 & 1 & 0 & 0 \\
1 & 1 & 1 & -1 & 0 & 0 \\
-1 & 1 & 0 & 0 & -1 & 1 \\
1 & -1 & 0 & 0 & -1 & 1
\end{pmatrix}\\
V_1 &= \frac{1}{\sqrt{2}}
\begin{pmatrix}
0 & 0 & 0 & 0 & 1 & 1 \\
0 & 0 & 0 & 0 & -1 & 1 \\
0 & 0 & 1 & 1 & 0 & 0 \\
0 & 0 & -1 & 1 & 0 & 0 \\
1 & 1 & 0 & 0 & 0 & 0 \\
-1 & 1 & 0 & 0 & 0 & 0
\end{pmatrix},
&\qquad
V_2 &=
\begin{pmatrix}
0 & 0 & \tfrac{1}{2} & \tfrac{1}{2} & \tfrac{1}{2} & \tfrac{1}{2} \\
0 & 0 & -\tfrac{1}{2} & -\tfrac{1}{2} & \tfrac{1}{2} & \tfrac{1}{2} \\
0 & 0 & -\tfrac{1}{2} & \tfrac{1}{2} & -\tfrac{1}{2} & \tfrac{1}{2} \\
\tfrac{1}{\sqrt{2}} & \tfrac{1}{\sqrt{2}} & 0 & 0 & 0 & 0 \\
-\tfrac{1}{\sqrt{2}} & \tfrac{1}{\sqrt{2}} & 0 & 0 & 0 & 0 \\
0 & 0 & \tfrac{1}{2} & -\tfrac{1}{2} & -\tfrac{1}{2} & \tfrac{1}{2}
\end{pmatrix}
\\[1.5em]
V_3 &=
\begin{pmatrix}
0 & 0 & \tfrac{1}{2} & \tfrac{1}{2} & \tfrac{1}{2} & \tfrac{1}{2} \\
0 & 0 & -\tfrac{1}{2} & -\tfrac{1}{2} & \tfrac{1}{2} & \tfrac{1}{2} \\
0 & \tfrac{1}{\sqrt{2}} & -\tfrac{1}{2} & \tfrac{1}{2} & 0 & 0 \\
\tfrac{1}{\sqrt{2}} & 0 & 0 & 0 & -\tfrac{1}{2} & \tfrac{1}{2} \\
0 & \tfrac{1}{\sqrt{2}} & \tfrac{1}{2} & -\tfrac{1}{2} & 0 & 0 \\
\tfrac{1}{\sqrt{2}} & 0 & 0 & 0 & \tfrac{1}{2} & -\tfrac{1}{2}
\end{pmatrix}
& &
\end{aligned}
\]

\end{document}